\documentclass{amsart}
\usepackage{hyperref}
\usepackage{amssymb}
\usepackage[foot]{amsaddr}
\usepackage{amsmath,amsthm}
\usepackage{hyperref}
\usepackage{paralist}
\usepackage{comment}
\usepackage{tikz}
\usetikzlibrary{positioning,shapes,shadows,arrows}
\hypersetup{
pdftitle={Faster algorithms  for vertex partitioning problems
  parameterized by clique-width}, 
pdfauthor={Sang-il Oum, Sigve Hortemo S\ae{}ther, Martin Vatshelle}
}

\usepackage{color}

\newcommand{\compl}[1]{\overline{#1}}
\newcommand{\bigoh}[1]{\mathcal{O}\hspace{-1pt}\left(#1\right)}
\newcommand{\bigohstar}[1]{#1\poly}
\newcommand{\ark}{$\mathbb{Q}$-cut-rank}
\newcommand{\arkw}{$\mathbb{Q}$-rank-width}
\newcommand{\Rrw}{\operatorname{rw}_\mathbb{Q}}%
\newcommand{\rw}{\operatorname{rw}}
\newcommand{\crk}{\operatorname{cutrk^{\mathbb{Q}}}}
\newcommand{\nec}{nec}
\newcommand{\necw}[1][d]{$\nec_{#1}$-width}
\newcommand{\cut}[1]{(#1,\compl{#1})}

\newcommand{\cw}{\operatorname{cw}}
\newcommand{\mult}{\cdot} 

\newcommand{\braced}[1]{\left\{#1\right\}}
\newcommand\abs[1]{\lvert #1\rvert}

\newcommand{\cclass}[1]{\textnormal{\textsf{#1}}}
\newcommand{\pname}[1]{\textnormal{\textsc{#1}}}
\newcommand{\msol}[1]{$MSOL_{1}$}

\newtheorem{theorem}{Theorem}[section]
\newtheorem{lemma}[theorem]{Lemma}
\newtheorem{open}{Open Problem}

\newtheorem{corollary}{Corollary}[theorem]

\newcommand{\poly}[1][n]{{#1}^{\bigoh{1}}}

\begin{document}

\title[Faster algorithms parameterized by clique-width]%
{Faster Algorithms %
  for Vertex Partitioning Problems
  Parameterized by Clique-width}

\author{Sang-il Oum}
\address[Oum]{Department of Mathematical Sciences, KAIST, Daejeon, South Korea}
\author{Sigve Hortemo S\ae{}ther}
\author{Martin Vatshelle}
\address[S\ae{}ther and Vatshelle]{Department of Informatics, University of Bergen, Norway}
\thanks{The first author is supported by Basic Science Research
  Program through the National Research Foundation of Korea (NRF)
  funded by  the Ministry of Science, ICT \& Future Planning
  (2011-0011653).
The second and third authors are supported by the Research Council of Norway.}
\email{sangil@kaist.edu}
\email{sigve.sether@ii.uib.no}
\email{vatshelle@ii.uib.no}

\begin{abstract}
 Many \cclass{NP}-hard problems, such as \textsc{Dominating Set}, are
 \cclass{FPT} parameterized by clique-width.
 For graphs of clique-width $k$ given with a $k$-expression, 
\textsc{Dominating Set} can be solved in $4^k \poly$ time. 
 However, no \cclass{FPT} algorithm is known for computing an optimal $k$-expression.
 For a graph of clique-width $k$, if we rely on known algorithms to
 compute a $(2^{3k}-1)$-expression via rank-width and then solving \textsc{Dominating Set} using the $(2^{3k}-1)$-expression, 
 the above algorithm will only give  %
 a runtime of $4^{2^{3k}} \poly$.
 There have been results which overcome this exponential jump; the best
 known algorithm can solve \textsc{Dominating Set} in time
 $2^{\bigoh{k^2}} \poly$ by avoiding constructing a $k$-expression 
 [Bui-Xuan, Telle, and Vatshelle.
\newblock Fast dynamic programming for locally checkable vertex subset and
  vertex partitioning problems.
\newblock {\em Theoret. Comput. Sci.}, 2013.
\newblock doi:
  \href{http://dx.doi.org/10.1016/j.tcs.2013.01.009}{10.1016/j.tcs.2013.01.009}].
 We improve this to $2^{\bigoh{k\log k}}\poly $. Indeed, 
 we show that for a graph of clique-width $k$, a large class of domination and partitioning problems (LC-VSP), 
 including \textsc{Dominating Set}, can be solved in $2^{\bigoh{k\log{k}}} \poly$.
 Our main tool is a variant of rank-width using the rank of a $0$-$1$ matrix over the rational field instead of the binary field.
\end{abstract}

\maketitle

\section{Introduction} \label{sec:intro}

 Parameterized complexity is a field of study dedicated to solving \cclass{NP}-hard problems efficiently on restricted inputs, and has grown to become a well known field over the last 20 years.
 Especially the subfields of Fixed Parameter Tractable (\cclass{FPT}) algorithms and kernelizations have attracted the interest of many researchers.
 Parameterized algorithms measure the runtime in two parameters; the input size $n$ and a secondary measure $k$ (called a parameter, either given as part of the input or being computable from the input).
 An algorithm is \cclass{FPT} if it has runtime $f(k) \poly$. Since we study \cclass{NP}-hard problems, we must expect that $f(k)$ is exponentially larger than $n$ for some instances. However, a good parameter is one where $f(k)$ is polynomial in $n$ for a large class of inputs.
 For a survey on parameterized complexity and \cclass{FPT}, we refer the reader to \cite{flum2006parameterized,downey1999parameterized,niedermeier2006invitation}.

 The \emph{clique-width} of a graph $G$, introduced by Courcelle and Olariu~\cite{CO2000}, is the minimum $k$ such that $G$ can
 be expressed by a $k$-expression, where a $k$-expression is an algebraic
 expression using the following four operations:
 \begin{itemize}
 \item $i(v)$: construct a graph consisting of a single vertex with label $i\in\{1,2,\ldots,k\}$.
 \item $G_1 \oplus G_2$: take the disjoint union of labelled graphs $G_1$ and $G_2$.
 \item $\eta_{i,j}$ for distinct $i,j\in \{1,2,\ldots,k\}$: add an
   edge between each vertex  of label $i$ and each vertex of label $j$.
 \item $\rho_{i \to j}$ for $i,j\in\{1,2,\ldots,k\}$: relabel each
   vertex of label $i$ to $j$.
 \end{itemize}
 Clique-width is a well-studied parameter in parameterized complexity
 theory. It is therefore interesting to be able to expand our
 knowledge on the parameter and to improve on the preciseness of
 problem complexity when parameterizing by clique-width.

 Courcelle, Makowsky, and
 Rotics~\cite{courcelle2000linear} showed that, for an input graph of
 clique-width at most $k$, every problem expressible in \msol1
 (monadic second-order logic of the first kind) can be solved in
 \cclass{FPT} time parameterized by $k$ if a
 $k$-expression %
 for the graph %
 is given together with the input graph.  Later, Oum and
 Seymour~\cite{oum2006approximating} gave an algorithm to find a
 $(2^{3k+2}-1)$-expression of a graph having clique-width at most $k$
 in time $2^{3k} \poly$.\footnote{Later, Oum~\cite{Oum2006} obtained an improved algorithm to find a
 $(2^{3k}-1)$-expression of a graph having clique-width at most $k$ in
 time $2^{3k}\poly$.}
 By combining these results, we deduce that
 for an input graph of clique-width at most $k$, every \msol1 problem
 is in \cclass{FPT}, even if a $k$-expression is not given as an input.
 However the dependency in $k$ is huge and can not be considered of
 practical interest.  In order to increase the practicality of
 \cclass{FPT} algorithms, it is very important to control the runtime
 as a function of $k$.

 If we rely on finding an approximate $k$-expression first and then doing dynamic programming on the obtained $k$-expression, we have two ways to make improvements; either we improve the algorithm that uses the $k$-expression, or we find a better approximation for clique-width.
 Given a $k$-expression, \textsc{Independent Set} and \textsc{Dominating Set} can be solved in time $2^k \poly$~\cite{Gurski08} and $4^k \poly$~\cite{BLRV10}, respectively. 
 Lokshtanov, Marx and Saurabh \cite{LMS11} show that unless the Strong ETH fails\footnote{%
  The Strong Exponential Time Hypothesis (Strong ETH) states that
  \textsc{SAT} can not be solved in $\bigoh{(2-\epsilon)^n}$ time for
  any constant $\epsilon > 0$. Here $n$ denotes the number of
  variables.
}, \textsc{Dominating Set} can not be solved in $(3-\epsilon)^k \poly$
time even if a $k$-expression is given\footnote{Their proof uses pathwidth, but the statement holds since clique-width is at most 1 higher than pathwidth.}.
 Hence, there is not much room for improvement in the existing algorithms when a $k$-expression is given.

 There are no known \cclass{FPT} algorithms for computing optimal
 $k$-expressions, and the best known \cclass{FPT} algorithm for
 approximating an optimal $k$-expression via rank-width has an approximation ratio which is exponential in the
 optimal clique-width~\cite{Oum2006}. Therefore, even for the simple
 \cclass{NP}-hard problems such as \textsc{Independent Set} and
 \textsc{Dominating Set}, all known algorithms following this
 procedure has a runtime where the dependency is double exponential in
 the clique-width.  The question of finding a better approximation algorithm for
 clique-width is an important and challenging open problem.

 However, there is a way around this by avoiding a $k$-expression: Bui-Xuan, Telle and
 Vatshelle~\cite{BTV09_hjoin} showed that by doing dynamic programming
 directly on a rank decomposition, \textsc{Dominating Set} can be
 solved in $2^{k^2} \poly$ for graphs of clique-width $k$.
 Their algorithm with  a runtime of $2^{\bigoh{k^2}}  \poly$ is not only for \textsc{Independent Set} and \textsc{Dominating
   Set} but also for a wide range of problems, called the \emph{locally checkable vertex subset and partitioning problems} (LC-VSP problems).
 Tables \ref{sigma-rho-table} and \ref{Dq-table} list some well known problems in LC-VSP.

\begin{table}
  \centering  
  \begin{tabular}{c|l}   \hline
    $d(\pi)$ & Standard name \\ 
    \hline %
    $d$       &  \pname{$d$-Dominating set}  \\
    $ d + 1 $ & \pname{Induced $d$-Regular Subgraph}  \\
    $d$       &  \pname{Subgraph of Min Degree $\geq d$}  \\
    $ d + 1 $ & \pname{Induced Subgraph of Max Degree $\leq d$} \\
    $2$       & \pname{Strong Stable set} or \pname{2-Packing} \\
    $2$       & \pname{Perfect Code} or \pname{Efficient Dominating set} \\
    $2$       & \pname{Total Nearly Perfect set}  \\
    $2$       & \pname{Weakly Perfect Dominating set} \\
    $2$       & \pname{Total Perfect Dominating set} \\
    $2$       & \pname{Induced Matching} \\ 
    $2$       & \pname{Dominating Induced Matching} \\
    $2$       & \pname{Perfect Dominating set} \\
    $1$       & \pname{Independent set} \\ 
    $1$       & \pname{Dominating set} \\
    $1$       & \pname{Independent Dominating set} \\
    $1$       & \pname{Total Dominating set}   \\ 
    \hline
  \end{tabular}
  \caption{A table of some vertex subset properties whose
    optimization problems belong to LC-VSP. 
    The meaning of the problem specific constant $d(\pi)$ is discussed in subsection \ref{subsec:LC-VS}.}
  \label{sigma-rho-table}
\end{table}%

\begin{table}
  \centering

\begin{tabular}{c|l}
 \hline
 $d(\pi)$ & Standard name\\ 
 \hline %
 $1$ & \pname{$H$-coloring} or \pname{$H$-homomorphism}\\
 $1$ & \pname{$H$-role assignment} or \pname{$H$-locally surjective homomorphism}\\
 $2$ & \pname{$H$-covering} or \pname{$H$-locally bijective homomorphism}\\
 $2$ & \pname{$H$-partial covering} or \pname{$H$-locally injective homomorphism}\\
 \hline
\end{tabular}
  \caption{A table of some homomorphism problems in LC-VSP for fixed simple graph $H$.
    These are expressible with a degree constraint matrix $D_q$ where $q(\pi)=\abs{V(H)}$. 
    The meaning of $D_q$, $d(\pi)$ and $q(\pi)$ is explained in
    subsection \ref{subsec:LC-VS}.}
\label{Dq-table}

\end{table}

 In this paper we improve on these results by using a slightly
modified definition of rank-width, 
called \emph{\arkw}, based on the rank function over the
rational field instead of the binary field.
 The idea of using fields other than the binary field for rank-width
 was investigated earlier in \cite{KR12}, but our work is the first to
 use \arkw{} to speed up an algorithm.

 We will show the following:
 \begin{itemize}
  \item For any graph, its \arkw{} is no more than its clique-width.
  \item There is an algorithm to find a decomposition confirming that
    \arkw{} is at most $3k+1$ for graphs of \arkw{} at most $k$ in
    time $\bigohstar{2^{3k}}$.
  \item If a graph has \arkw{} at most $k$, then %
    every fixed LC-VSP problem can be solved in
    $\bigohstar{2^{\bigoh{k\log k}}}$-time.
 \end{itemize}
 This allows us to construct an algorithm that runs in time
 $\bigohstar{2^{\bigoh{k\log k}}}$ for graphs of clique-width at most
 $k$ and solve every fixed LC-VSP problem, improving the previous runtime $2^{\bigoh{k^2}}\poly$
of the
 algorithm by Bui-Xuan et al.~\cite{BTV13}.

 We also relate the parameter \arkw{} to other existing parameters.
 There are several factors affecting the quality of a parameter, such as: 
  Can we compute or approximate the parameter? Which problems can we solve in \cclass{FPT}\cclass{} time? Can we reduce the exponential dependency in the parameter for specific problems? And, how large and natural is the class of graphs having a bounded parameter value?

 This paper is organized as follows: 
 In Section~\ref{sec:framework} we introduce the main parts of the framework used by Bui-Xuan et al.~\cite{BTV13}, including the general algorithm they give for LC-VSP problems.
 Section~\ref{sec:our_parameter} revolves around \arkw\ and is where the results of this paper reside.
 We show how \arkw\ relates to clique-width, and reveal why we have a good \cclass{FPT} algorithm for approximating a decomposition.
 In Section~\ref{sec:relation_Qrw_nec}, we give our main result, which is an improved upper bound on solving LC-VSP problems parameterized
 by clique-width when we are not given a decomposition.
 We end the paper with Section~\ref{sec:conclusion} containing some concluding
 remarks and open problems.

\section{Framework}\label{sec:framework}

We write $V(G)$ and $E(G)$ to denote the set of
 vertices and edges, respectively, of a graph $G$.
 For $A \subseteq V(G)$, let  $\overline{A}=V(G)\setminus A$.
 For a vertex $v \in V(G)$, let $N_G(v) $ be the set of all neighbours
 of $v$ in $G$. We  omit the subscript if  it is clear from the context.
 For a set $S \subseteq V(G)$ we define $N(S) = \bigcup_{v \in S} N(v) \setminus S$.

\subsection{Branch Decompositions}
\label{subsec:branch-decompositions}

 The algorithm of Bui-Xuan et al.~\cite{BTV13} needs a branch decomposition as input.
  A \emph{branch decomposition} $(T, \delta)$ of a graph $G$ consists of a subcubic tree $T$ (a tree of maximum degree $3$) and a bijective function $\delta$ from the leaves of $T$ to the vertices of $G$. (Note that this definition differs from that of \cite{robertson1991graph} by $\delta$ mapping to the vertices of $G$ instead of the edges.)

 Every edge in a tree splits the tree into two connected components.
 In a branch decomposition $(T, \delta)$ for a graph $G$, we say that
 each edge $e$ of $T$ induces a \emph{cut} in $G$.
 This induced cut is a bipartition $(A,\compl{A})$ of the vertices of $V(G)$ so that $A$ is the set of vertices mapped by $\delta$ from vertices of one component of $T - e$, and $\overline{A}$ is the set of vertices mapped by $\delta$ from the other component of $T-e$ (see Figure~\ref{fig:a_cut}).

 \begin{figure}[h!]
   \centering
   \includegraphics{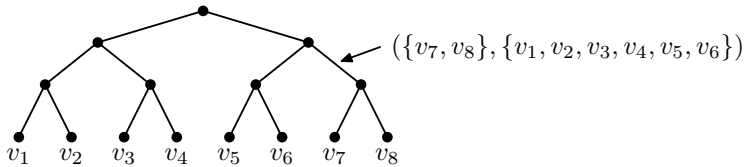}
   \caption{A branch decomposition $(T, \delta)$ of a graph with $8$
     vertices. The leaves of $T$ is mapped by $\delta$ to each of the
     vertices. Each edge in $T$ induce a cut.}
   \label{fig:a_cut}
 \end{figure}

 A function $f:2^{V(G)} \rightarrow \mathbb{R}$ is called a \emph{cut-function} if it is symmetric, that is $f(A)= f(V(G) \setminus A)$ for all $A\subseteq V(G)$.  As
 there is a bijection from each subset $A \subseteq V(G)$ to each cut
 $\cut{A}$, we may abuse notation slightly and say that $f$ is also
 defined for cuts of $V(G)$, by regarding $f\cut{A}$ as $f(A)$.

 Given a cut-function $f$ and a branch decomposition $(T, \delta)$ of a graph $G$, 
 \begin{itemize}
  \item the \emph{$f$-width} of $(T, \delta)$ is the maximum value of $f$ over all the cuts of $(T, \delta)$, and
  \item the \emph{$f$-width} of $G$ is the minimum $f$-width over all possible branch decompositions of $G$.
  \end{itemize}
If $\lvert V(G)\rvert\le 1$, then $G$ admits no branch decomposition and we
define its $f$-width to be $f(\emptyset)$.

Many width parameters of graphs can be defined in terms of $f$-width for some cut-function $f$. For example, in a graph $G$, if we define
$f(A)$ to be the number of maximal independent sets of the subgraph of $G$ induced by edges having one end in $A$ and the other end in $\compl{A}$,
then the $f$-width is exactly the boolean-width~\cite{DBLP:journals/tcs/Bui-XuanTV11}.

 When we speak of the $f$-width of a graph, we address it as a \emph{width parameter} of the graph.

\subsection{Neighbourhood Equivalence}
\label{subsec:nec_d}

Two sets of vertices $S_1,S_2$ are \emph{neighbourhood equivalent} if
they have the same set of neighbours, in other words,  $N(S_1) = N(S_2)$.
 We are particularly interested in neighbourhood equivalence in
 bipartite graphs, or more specifically, cuts defined by a branch decomposition.
 This concept was generalized with respect to cuts in~\cite{BTV13}.
 We define the \emph{$d$-neighbour equivalence} relation $\equiv_A^d$, and use this to define the parameter $\nec_d$. 

 For a cut $\cut{A}$ of a graph $G$, and a positive integer $d$, two
 subsets $X, Y \subseteq A$ are \emph{$d$-neighbour equivalent}, $X
 \equiv_A^d Y$, over $\cut{A}$ if:
  \[
   \text{for each vertex } v \in \compl{A}, \quad 
   \min{\braced{d,\; \lvert N(v) \cap X\rvert }} = 
   \min{\braced{d,\; \lvert N(v) \cap Y\rvert}}.
  \]

  \begin{figure}[h!]
    \centering
    \includegraphics{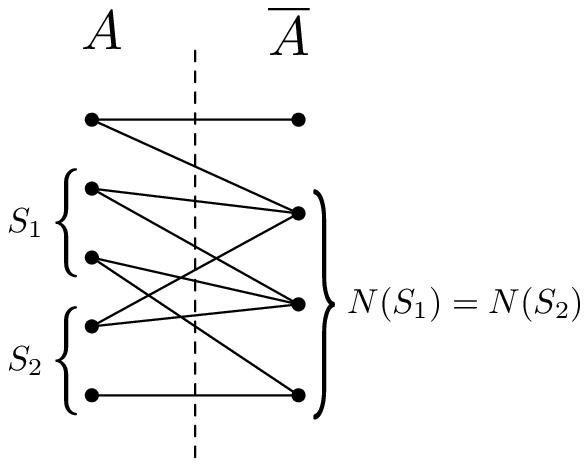}
    \caption{The sets $S_1$ and $S_2$ are neighbourhood equivalent over $\cut{A}$. That is, $S_1 \equiv_A^1 S_2$. However, in this example it is not the case that $S_1 \equiv_A^2 S_2$.}
    \label{fig:neighbouhoodEquivalence}
  \end{figure}

 The \emph{number of $d$-neighbour equivalence classes}, $\nec_d(A)$, is the number of equivalence classes of $\equiv_A^d$ over $\cut{A}$.

 In other words, $X \equiv_A^d Y$ over the cut $\cut{A}$ if each vertex in $\compl{A}$ is either adjacent to at least $d$ vertices in both $X$ and $Y$, 
 or is adjacent to exactly the same number of vertices in $X$ as in $Y$.
 The algorithm in \cite{BTV13} uses this relation to limit the number
 of partial solutions to try.
 Therefore, the runtime is dependent on the number of $d$-neighbour equivalence classes.

\subsection{Locally Checkable Vertex Subset and Vertex Partitioning Problems} \label{subsec:LC-VS}

Telle and Proskurowski~\cite{DBLP:journals/siamdm/TelleP97} introduced 
the \emph{Locally Checkable Vertex Subset and Vertex
  Partitioning problems} (LC-VSP), 
 also called $[\sigma, \rho]$-problems and $D_q$-partition problems.
 This is a framework to describe many well-known
 graph problems, see~\cite{DBLP:journals/siamdm/TelleP97,BTV13}. 
 Tables \ref{sigma-rho-table} and \ref{Dq-table} list
 some of them. For completeness, we give the definitions of the
 problem class LC-VSP, however, they are not used directly in this
 paper and can be skipped by the reader.

 For finite or co-finite sets $\sigma$ and $\rho$ of non-negative
 integers, a set~$S$ of vertices of a graph $G$ is a \emph{$[\sigma,
   \rho]$-set} of $G$ if for each vertex $v$ of $G$,
 \[
   \left\lvert N(v) \cap S \right\rvert \in 
   \begin{cases}
     \sigma & \text{if } v \in S, \\
     \rho & \text{if } v \in V(G) \setminus S.
   \end{cases}
 \]
 The \emph{Locally Checkable Vertex Subset problems} (LC-VS), or 
 \emph{$[\sigma,\rho]$-problems}, are those problems
 that consist of finding a minimum or maximum $[\sigma, \rho]$-set of
 the input graph.
 
 The \emph{LC-VSP problems}, or \emph{$D_q$-partition problems}, is a generalization
 of the LC-VS problems.
 A \emph{degree constraint matrix} $D_q$ is a $q \times q$ matrix such
 that each
 cell is a finite or co-finite set of non-negative integers. 
 We say that a
 partition $V_1, V_2, \ldots, V_q$ of $V(G)$
 satisfies $D_q$ if for $ 1 \leq i,j \leq q$,
 the number of neighbours in $V_j$ of a vertex of $V_i$
 is in the set $D_q[i,j]$.
 In other words, 
 \[ 
 \left\lvert N(v) \cap V_j \right\rvert \in D_q[i,j]
 \mbox{ for all $1\le i,j\le q$ and $v\in V_i$.}
 \]
 For a given degree constraint matrix $D_q$, the LC-VSP problem is to decide whether the vertex set of a graph admits a partition satisfying $D_q$.

 For each LC-VSP-problem $\pi$, there are two problem-specific
 constants $d(\pi)$ and $q(\pi)$. The number $q(\pi)$ equals the
 number of parts in a partition that the problem requests, or
 equivalently, the row/column size of the constraint matrix (i.e., for
 problem $\pi$ with degree constraint matrix $D_q$ we have $q =
 q(\pi)$.  The number $d(\pi)$ is defined to be one more than the
 largest number in all the finite sets and in all the complements of
 the co-finite sets of the degree constraint matrix used for
 expressing $\pi$. If all the finite sets and complements of co-finite
 sets are empty, $d(\pi)$ is zero.  

 For example, \textsc{Dominating Set} can be described by a degree
 constraint matrix $D_2$ where $D_2[1,1] = D_2[1,2] = D_2[2,2] =
 \mathbb{N}$ and $D_2[2,1] = \mathbb{N} \setminus \{0\}$, and we ask
 to minimize $\abs{V_1}$. %
 If we alter
 $D_2[1,1]$ to $\{0\}$, the problem is changed to \textsc{Independent Dominating Set}, as no vertex in $V_1$ can be adjacent to another
 vertex in $V_1$. For both problems we have a $2\times
 2$-matrix, and so $q(\pi) = 2$ and $d(\pi) = 0 + 1 = 1$.

 The algorithm of Bui-Xuan et al.~\cite{BTV13} solves each of the
 LC-VSP problems with a runtime dependent on \necw[d(\pi)] and $q(\pi)$
 by using the $d$-neighbour equivalence relation $\equiv_A^{d(\pi)}$.
 
 \begin{theorem}[Bui-Xuan et al.~{\cite[Theorem 2]{BTV13}}] \label{thm:LC-VSgivenNEC} 
   Let $\pi$ be a problem in LC-VSP. For a graph $G$ given 
   with its branch decomposition of \necw[d(\pi)] $k$,
   the problem~$\pi$ can be solved in time 
   $\bigoh{\abs{V(G)}^4  \mult q(\pi)\mult  k^{3q(\pi)}}$.
 \end{theorem}

\section{\texorpdfstring{\arkw}{Q-rank-width} of a Graph} \label{sec:our_parameter}

The \emph{\ark} function of a graph $G$ is a function on the subsets
of $V(G)$ that maps $X \subseteq V(G)$ to the rank of an $\abs{X} \times
\left\lvert\compl{X}\right\rvert$-matrix $A=(a_{ij})_{i\in X, j\in
  \compl{X}}$ over the rational field such that $a_{ij}=1$ if $i$ and
$j$ are adjacent in $G$ and $a_{ij}=0$ otherwise. We let
\emph{$\crk(X)$} denote the \ark\ of $X$.  For a subset $X \subseteq
V(G)$, the matrix $A$ associated with $\crk(X)$ is the \emph{adjacency
  matrix} of the cut $\cut{X}$.  Note that if the underlying field of
the matrix $A$ is the binary field $GF(2)$, then we obtain the
definition of the usual cut-rank function \cite{oum2006approximating}.
By \emph{\arkw{}} of a graph, we mean its \ark{}-width (see subsection
\ref{subsec:branch-decompositions}). We may denote the \arkw\ simply
as $\Rrw$.

\begin{figure}[!h]
  \centering
  \[
  \bordermatrix{ & & \compl{A} &  &\cr
    & 1 & 1&  0 & 0\cr
    & 0 &  1& 1 & 0 \cr
    A & 0& 0&1&1\cr
    & 0&1&1&0\cr
    &0&0&0&1}
  \]
  \caption{The adjacency matrix of the cut depicted in Figure~\ref{fig:neighbouhoodEquivalence}. The rank over the rational field is $4$, so the \ark\ of this cut is $4$.}
\end{figure}

Since the \ark{} function is symmetric submodular and is computable in
polynomial time, by applying the result of Oum and
Seymour~\cite{oum2006approximating}, we get the following theorem.

\begin{theorem}[Oum and Seymour~\cite{oum2006approximating}]
 \label{thm:approximate-Q-rw} 
 There is a $\bigohstar {2^{3k}}$-time algorithm for which, given a graph $G$ as input and a parameter $k$, either outputs a branch decomposition for $G$ of \arkw{} at most $3k+1$ or confirms that \arkw{} of $G$ is more than $k$.
\end{theorem}

\subsection{$\mathbb Q$-rank-width versus clique-width/rank-width}
 
 The question of how useful the \arkw{} is as a width parameter is hard to answer.
 To better understand this question, it would be interesting to know the relation to other well-known width parameters such as treewidth, rank-width and clique-width.

 The following relates \arkw{} to the closely related parameter rank-width, yet we see that rank-width can be substantially lower than \arkw{}. 
 
\begin{lemma} \label{lemma:NewParamLEQcw}
For any graph $G$ we have $\rw(G) \leq \Rrw(G) \leq \cw(G) \leq 2^{\rw(G)+1}-1.$
\end{lemma}
\begin{proof}
  The first inequality is from the fact that a set of $0$-$1$ vectors
  linearly dependent over $\mathbb{Q}$ must also be linearly
  dependent over $GF(2)$. 

  The second and third inequalities follow from \cite[Proposition
  6.3]{oum2006approximating} since their proof is not dependent on the
  type of field rank-width uses. They show that a $k$-expression can
  be translated to a branch decomposition where for every cut
  $\cut{A}$ in the decomposition, either the number of distinct rows or
  the number of distinct columns in the adjacency matrix $M$ of its
  induced bipartite graph, is bounded by $k$.  Since this means the
  rank of $M$ over $\mathbb{Q}$ is at most $k$, we have $\Rrw(G) \leq
  \cw(G)$. The idea of showing $\cw(G) \leq 2^{\rw(G)+1}-1$, is that a
  branch decomposition where the adjacency matrix of each cut has its
  number of distinct columns/rows (approximately) bounded by some $k$,
  can be translated to a $k$-expression. As the number of distinct
  columns/rows for any $0$-$1$ matrix of rank $\rw$ is at most
  $2^{\rw}$, we get our inequality. The last two inequalities are also
  proved in~\cite{KR12}.  
\end{proof}

We believe Lemma~\ref{lemma:NewParamLEQcw} is tight. There are
existing results showing that it is almost tight.  A $n \times n$ grid
has rank-width $n-1$~\cite{Jel10} and clique-width $n+1$~\cite{GR00},
hence the first two inequalities are almost tight.  There exist graphs
with treewidth $k$ and hence $\mathbb{Q}$-rank-width at most $ k$ and
clique-width at least $2^{\lfloor k/2 \rfloor-1}$~\cite{CR05}.

\subsection{$\mathbb Q$-rank-width versus treewidth/branch-width}
Oum~\cite{Oum2008} proved that the rank-width of a graph is less than
or equal to its tree-width plus $1$.
We prove a similar result for \arkw{}.
In order to show this, we use the notion of \emph{tangles} and
\emph{branch-width} of symmetric submodular functions, see
\cite{GGRW2003,OS2005}.
For a symmetric submodular function $f$ on a finite set $V$, an
\emph{$f$-tangle} $\mathcal T$ of order $k+1$ is a set
of subsets of $V$ satisfying the following:
\begin{enumerate}[(T1)]
\item\label{t1} For all $A\subseteq V$, 
  if $f(A)\le k$, then either $A\in \mathcal T$ or $V\setminus
  A\in \mathcal T$.
\item\label{t2} If $A, B, C\in\mathcal T$, then $A\cup B\cup C\neq V$.
\item\label{t3} For all $v\in V$, we have $V\setminus\{v\}\notin
  \mathcal T$.
\end{enumerate}
\begin{theorem}[Robertson and Seymour
  {\cite[(3.5)]{robertson1991graph}}; Geelen, Gerards, Robertson, and
  Whittle~{\cite[Theorem  3.2]{GGRW2003}}]\label{thm:tangle} 
  There is no $f$-tangle of order $k+1$
  if and only if 
  the branch-width of $f$ is at most $k$.
\end{theorem}
  For a set $X$ of edges, let $T_X$ be the set of vertices 
  incident with at least one of the edges in $X$.
  For a set $X$ of edges, let $\eta(X)=\abs{T_X\cap T_{E(G)\setminus X}}$,
  that is the number of vertices incident to both edges in $X$ and
  edges in $E(G)\setminus X$.
  Then the \emph{branch-width} of a graph $G$ is the branch-width of  the
  function $\eta$ on $E(G)$ \cite{robertson1991graph}.

\begin{lemma}\label{lem:etarho}
  For $X\subseteq E(G)$, 
  we have $\crk(T_X)\le \eta(X)$.
\end{lemma}
\begin{proof}
  Suppose $k=\eta(X)$. Then $T_X$ has at most $k$ vertices having
  neighbors in $V(G)\setminus T_X$ by the definition of $\eta$. Thus
  $\crk(T_X)\le k$ as the rank of a matrix with at most $k$ non-zero
  rows is at most $k$.
\end{proof}
\begin{lemma}\label{lem:bw}
  Let $k\ge 2$. If $G$ has \arkw{} at least $k+1$, then 
  $G$ has branch-width at least $k+1$.
\end{lemma}
\begin{proof}
  We may assume that $G$ is connected without loss of generality.
  Let $\rho$ be the \ark{} function of $G$.
  Since the \arkw{} of $G$ is larger than $k$, there exists 
  a $\rho$-tangle $\mathcal T$ of order $k+1$. 

  We aim to construct the tangle $\mathcal U$ of order $k+1$ as
  follows.
  Let \[\mathcal U=\{ X\subseteq E(G): 
  \eta(X)\le k,~ T_X\in \mathcal T\}.\] We claim that $\mathcal U$ is   an
  $\eta$-tangle of order $k+1$.

  (1) Suppose that $\eta(X)\le k$ for a set $X$ of edges. 
  We need to show that either $X\in \mathcal U$ or $E(G)\setminus
  X\in \mathcal U$. Suppose that $X\notin \mathcal U$ and
  $E(G)\setminus X\notin \mathcal U$. Then, 
  $T_X\notin \mathcal T$. Since $\rho(T_X)\le \eta(X)\le k$ and $T$ is a
  $\rho$-tangle, we know that 
  $V(G)\setminus  T_X\in \mathcal T$.
  Similarly we deduce that $V(G)\setminus T_{E(G)\setminus X}\in
  \mathcal T$.
  Moreover since $\eta(X)\le k$, $T_X\cap T_{E(G)\setminus X}\in
  \mathcal T$ (easy to show by induction---any set of at most $k$
  vertices 
  belongs to a $\rho$-tangle of order $k+1$).
  This leads a contradiction because
  $(V(G)\setminus T_X)\cup (T_X\cap T_{E(G)\setminus X})\cup
  (V(G)\setminus T_{E(G)\setminus X})=V(G)$
  and $\mathcal T$ is a $\rho$-tangle.
  
  (2) Suppose that $X\cup Y\cup Z=E(G)$ for three sets $X,Y,Z\in
  \mathcal U$.
  If $v\notin T_X\cup T_Y\cup T_Z$, then $v$ is an isolated
  vertex.
  Since $G$ is connected, there is no such $v$.
  Thus, $T_X\cup T_Y\cup T_Z=V(G)$ and 
  $T_X,T_Y,T_Z\in \mathcal T$. A contradiction.
  
  (3) For each edge $e$, $\eta(\{e\})\le 2$ and therefore
  if $k\ge 2$, then  $T_{\{e\}}\in \mathcal T$.
  So $\{e\}\in \mathcal U$.

  By (1)--(3), we checked all axioms for $\eta$-tangles.
\end{proof}
\begin{theorem}
$\Rrw(G)\le \max(\mbox{branch-width}(G),1)\le \mbox{treewidth}(G)+1$.
\end{theorem}
\begin{proof}
  If the branch-width of $G$ is larger than $1$, then by
  Lemma~\ref{lem:bw}, we know that the rank-width is at most the branch-width
  of $G$. 
  If the branch-width of $G$ is 1, then $G$ is a forest and therefore
  the rank-width is at most $1$.
  (But $G$ may have edges, even if branch-width of $G$ is 0
  and in this case, the rank-width of $G$ is 1.)
  Robertson and Seymour \cite{robertson1991graph} showed that 
  branch-width is at most  tree-width plus $1$.
\end{proof}
We remark that an identical proof can be used to show an analogous
result for variations of rank-width on different fields.

Figure~\ref{fig_compare} shows a comparison diagram of graph
parameters.  The idea of such a diagram is that parameterized
complexity results will propagate up and down in this diagram.
Positive results propagate upward; for instance, since \textsc{Dominating Set} is solvable in $2^{\bigoh{tw}} \poly$ for a graph of treewidth
$tw$~\cite{telle1993practical}, we see that \textsc{Dominating Set} is
solvable in $2^{\bigoh{pw}}\poly$ for a graph of pathwidth $pw$.
Negative results propagate downward; for example, since unless ETH
fails, \textsc{Dominating Set} can not be solved in $2^{o(pw)} \poly$
where $pw$ is the pathwidth of the input graph~\cite{LMS11eth}, so is
the case for treewidth, clique-width, \arkw, rank-width and
boolean-width. From this table, we can deduce that the
entire class LC-VSP cannot be in \cclass{FPT} parameterized by OCT,
D2Chordal or D2Perfect unless $\cclass{P} = \cclass{NP}$, since \textsc{Dominating Set} is \cclass{NP}-hard for both bipartite~\cite{Bertossi1986} and chordal
graphs~\cite{BJ1982}. Furthermore, LC-VSP parameterized by either of the remaining
parameters is in fact in \cclass{FPT}, since we have \cclass{FPT}
algorithms solving all problems in LC-VSP parameterized by rank-width,
boolean-width, and D2Interval\footnote{Given a graph of $\operatorname{D2Interval}(G)
  = k$, we have a fixed-parameter tractable algorithm to construct a branch deceomposition of $\nec_d$-width
  at most $2^kn^d$ and thus by the algorithm of \cite{BTV13},
  we have a fixed-parameter tractable algorithm to solve the problem.
  To do this, we first find a vertex
  set $S$ of size $k$ so that $G - S$ is an interval graph, and then
  find a branch decomposition of $\nec_d$-width at most $
  n^d$. Arbitrarily adding the vertices of $S$ anywhere in the branch
  decomposition cannot increase the $\nec_d$-width by more than
  $2^{|S|}$, and thus the resulting branch decomposition has
  $\nec_d$-width at most $2^kn^d$. We have a fixed-parameter tractable
  algorithm to find $S$ shown in \cite{DBLP:conf/soda/CaoM14} and constructing a branch
  decomposition for $G-S$ of $\nec_d$-width at most $n^d$ can be done in
  polynomial time by \cite{BV13}.}.

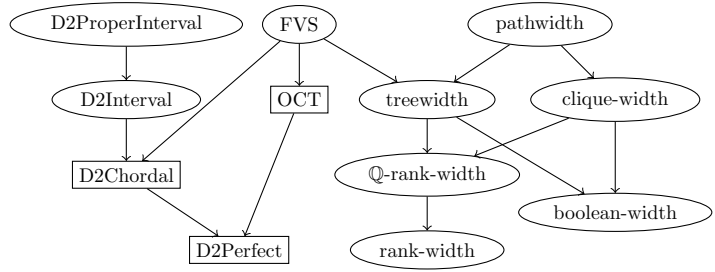
\begin{figure}
\begin{center}
\begin{tikzpicture}
\tikzstyle{every node}=[scale=.75,draw,ellipse,align = center]
\tikzstyle{every path}=[->]

\node (fvs)   at (-2.2,-1){FVS};
\node[rectangle] (oct)   at (-2.2,-2){OCT};
\node (pint)  at (-4.5,-1){D2ProperInterval};
\node (int)   at (-4.5,-2){D2Interval};
\node[rectangle] (cord)  at (-4.5,-3){D2Chordal};
\node[rectangle] (perf)  at (-3.0,-4){D2Perfect};
\node (boolw) at ( 2.0,-3.5){boolean-width};
\node (rw)    at (-0.5,-4){rank-width};
\node (cw)    at ( 2.0,-2){clique-width};
\node (pw)    at ( 1.0,-1){pathwidth};
\node (tw)    at (-0.5,-2){treewidth};
\node (qrw)   at (-0.5,-3){$\mathbb{Q}$-rank-width};

\draw (int) -- (cord);
\draw (pint) -- (int);
\draw (cord) -- (perf);
\draw (oct) -- (perf);      
\draw (fvs) -- (tw);
\draw (fvs) -- (cord);
\draw (pw) -- (tw);
\draw (tw) -- (qrw);
\draw (cw) -- (qrw);
\draw (fvs) -- (oct);
\draw (qrw) -- (rw);
\draw (tw) -- (boolw);
\draw (cw) -- (boolw);
\draw (pw) -- (cw);
\end{tikzpicture}
\caption{A comparison diagram of some graph parameters. 
  A parameter $\kappa_1$ is drawn below a parameter $\kappa_2$ if there is a
  constant $c$ such that  $\kappa_1(G) \le c\cdot {\kappa_2(G)}$ for all
  graphs $G$. The abbreviations are: FVS = Feedback Vertex Set number, OCT = Odd Cycle Transversal number, D2$\Pi$ = Vertex Deletion distance to a member of $\Pi$. 
  For the circled parameters all the LC-VSP problems are in \cclass{FPT}, and unless $\cclass{P} = \cclass{NP}$ for each of the remaining parameters at least one of the LC-VSP problems is not in \cclass{FPT}.}
\label{fig_compare}
\end{center}
\end{figure}

\section{Bounding \texorpdfstring{\necw{}}{nec-width} by
  \texorpdfstring{\arkw}{Q-rank-width} and its Algorithmic Consequences}
\label{sec:relation_Qrw_nec} 
Now we know how to find a branch decomposition with a low \arkw{}.
We are going to discuss its \necw\ to apply 
Theorem~\ref{thm:LC-VSgivenNEC}.
Theorem~\ref{thm:LC-VSgivenNEC} provides the runtime of the algorithm in terms of the
\necw\ of the given decomposition. So, if we manage to give a
bound on the \necw\ of a decomposition in terms of the \arkw{}, we will
also get a bound on the runtime of the algorithm in terms of
\arkw{}.  We will prove such a bound shortly, but in order to do this
we first need the following lemma, based on a proof of
Belmonte and Vatshelle \cite[Lemma 1]{BV13}.

\begin{lemma} \label{lemma:set-size}
  Given a positive integer $d$ and a cut $\cut{A}$ of \ark{} $k$, for every
  subset $S \subseteq A$, there exists a subset $R \subseteq S$ so that
  $\abs{R} \leq dk$ and $R \equiv_A^d S$ over the cut.
\end{lemma}

\begin{proof}
  We proceed by induction on $d$. If 
  $d = 1$, then let $S'$ be a minimal subset of $S$ so that $S'
  \equiv_A^1 S$. %
  Since $S'$ is minimal, removing any vertex of $S'$ will decrease
  $\abs{N(S')}$. Therefore, every vertex of $S'$ is adjacent to at
  least one vertex that none of the other vertices in $S'$ are
  adjacent to. In the adjacency matrix $M$ of $\cut{A}$, this means
  that each of the corresponding rows of $S'$ has a $1$ in a column
  where all the other rows of $S'$ has a $0$. Hence, the rows of $S'$
  are linearly independent and so $\abs{S'} \leq \crk(A) = k$.

  So we may assume that $d>1$.
  By the above, 
  there exists a subset $S_1 \subseteq S$ such that $\abs{S_1}\le k$ and
  $S_1 \equiv_A^1 S$. 
  By the induction
  hypothesis, there exists a set $S_2 \subseteq (S \setminus S_1)$ so
  that $S_2 \equiv_A^{d-1} (S \setminus S_1)$ and $\abs{S_2} \leq (d-1) k$.  

  We claim that $S_1\cup S_2 \equiv_A^{d} S$.
  Let $v\in \compl{A}$. 
  We may assume that $v$ has at most $d-1$ neighbours in $S_1\cup S_2$.

  If $v$ has a neighbour in $S\setminus (S_1\cup S_2)$, then 
  $\abs{ N(v) \cap  (S \setminus S_1) } > \abs{ N(v) \cap S_2 }$
  and therefore $v$ has at least $d-1$ neighbours in $S_2$ and so $v$ has
  no neighbors in $S_1$.
  This contradicts our assumption that $S_1\equiv_A^1 S$.

  Thus $v$ has no neighbour in $S\setminus (S_1\cup S_2)$. This proves
  the claim.
  Since $\abs{S_1\cup S_2}\le dk$, this completes the proof of the lemma.
\end{proof}

Lemma~\ref{lemma:set-size} implies that to count 
distinct $d$-neighbour equivalence classes for a cut
of a branch decomposition of \arkw\ $k$,
it is enough to 
search subsets of size at
most $dk$.
The same result is true, even if we replace \arkw\ with rank-width or
boolean-width (\cite{martinthesis}, \cite[Lemma 5]{BTV13}).

Then what is the contribution of \arkw\ instead of rank-width or boolean-width?
Here comes the crucial difference.
For both rank-width $k$ or boolean-width $k$, the number of vertices with
distinct neighbourhoods over the cut is no more than $2^k$
\cite{martinthesis,BTV13}.  Putting this together gives a trivial bound of
$\nec_d \leq 2^{d  k^2}$. 
We can improve this bound if $k$ is \arkw, 
thanks to the fact that the row space of
some matrix over $\mathbb Q$ not only contains all the rows of
the matrix, but also all the different sums of the rows in the
matrix. So, we can bound $\nec_d(A)$ by using a more direct connection between
\arkw\ and the number of distinct $d$-neighbourhoods than that of the trivial
bound.

\begin{theorem} \label{thm:NewParamToNEC}
  If the \arkw{} of a branch decomposition is $k$, then the \necw\ of
  the same decomposition is no more than $(dk+1)^{k}=2^{k \log_2 {(d  k+1)}}$.
\end{theorem}

\begin{proof}
  It is enough to prove that if a cut $\cut{A}$ has \ark{} $k$, 
  then 
  $\nec_d(A)\le (dk+1)^k$.
  Let $M$ be the $A\times \compl{A}$ adjacency matrix of the cut
  $\cut{A}$ over $\mathbb{Q}$. 

  For a subset $S$ of $A$, 
  let $\sigma(S)$ be the sum of the row vectors of $M$ corresponding to $S$. 
  If $\sigma(S) = \sigma(S')$ then $S \equiv_A^d S'$
  for all $d$, because the entries of $\sigma(S)$ represent the number
  of neighbours in $S$ for each vertex in $\compl{A}$.

  By Lemma~\ref{lemma:set-size}, each equivalence class of
  $\equiv_A^d$ can be represented by a subset $S$ of $A$ having at most
  $dk$ vertices.
  Notice that for such $S$, each entry of $\sigma(S)$ is in $\{0,1,2,\ldots,dk\}$.

  Let $B$ be a set of $k$ linearly independent columns of $M$.
  Since $M$ has rank $k$, 
  every linear combination of row vectors of $M$ is completely determined by its
  entries in $B$.
  Thus the number of possible values of $\sigma(S)$ is at most $(dk+1)^k$ (see Figure~\ref{fig:rowSum}).
      This proves the theorem.
\end{proof}

  \begin{figure}[h!]
    \centering
    \includegraphics{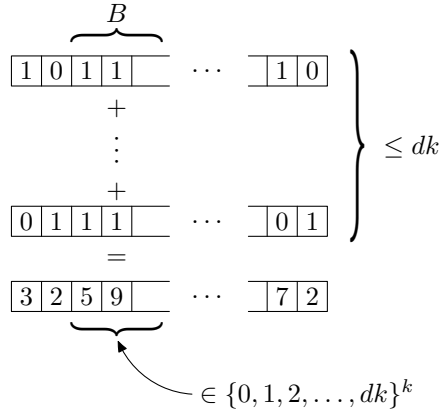}
    \caption{As described by Theorem~\ref{thm:NewParamToNEC}, we can determine the sum of the vectors by looking at the values in the columns $B$. As we sum over at most $dk$ rows, and each row either increases the value of a column by exactly one or exactly zero, the number of unique sums possible is at most $|\{0,1,\ldots\}|^{|B|} = (1+dk)^k$.}
    \label{fig:rowSum}
  \end{figure}

This result, combined with Theorems~\ref{thm:LC-VSgivenNEC} and
\ref{thm:approximate-Q-rw}, shows that all the LC-VSP problems can be
solved in time $\bigohstar{2^{\bigoh{k \log{k}}}}$. Expressing the
runtime in terms of clique-width, we get the following corollary.

\begin{corollary}
  Every LC-VSP problem $\pi$ on $n$-vertex graphs of clique-width $\cw$
  can be solved in $\bigohstar{2^{\bigoh{\cw \log{(\cw \mult{}d(\pi))}q(\pi)}}}$-time.
\end{corollary}
\begin{proof}
  Let $k$ be the \arkw{} of $G$.
  By Theorem~\ref{thm:approximate-Q-rw} we can find a
  branch decomposition of \arkw\ at most $3k+1$ in time $\bigohstar{2^{3k}}$.
 By Theorems \ref{thm:LC-VSgivenNEC} and \ref{thm:NewParamToNEC}, the
 LC-VSP problem $\pi$ can be solved in time $\bigohstar{2^{9k \log{(3k \mult d(\pi) + 1)}q(\pi)}}$.
 This completes the proof because $k \leq \cw$ by Lemma~\ref{lemma:NewParamLEQcw}.
\end{proof}

\section{Conclusion} \label{sec:conclusion} 
If we are given a $k$-expression as input, the best known \cclass{FPT}
algorithm parameterized by $k$ solving the \textsc{Dominating Set} is by
Bodlaender et al.~\cite{BLRV10} and runs in time $\bigohstar{4^{k}}$.
However, it is currently open whether we can construct a
$\bigoh{k}$-expression of an input graph of clique-width at most $k$ in
polynomial time.  We have shown the existence of algorithms with
runtime $2^{\bigoh{\cw \log{\cw}}} \poly$ for all LC-VSP problems,
without assuming that a $k$-expression is given as an input.  This
still leaves the natural open question:

\begin{open}
 Can \pname{Independent Set} or \pname{Dominating Set} be solved in $\bigohstar{2^{\bigoh{\cw}}}$ time, where $\cw$ is the clique-width of the graph?
\end{open}

We know that for a graph of treewidth $tw$, \textsc{Independent Set}
can be solved in time $2^{\bigoh{tw}} \poly$ time. This leads us to an interesting
question of what parameters give a single exponential runtime for
\pname{Independent Set}. Two such parameters are the Vertex Deletion
Distance to Proper Interval graphs (D2ProperInterval) and the Odd Cycle
Transversal number (OCT number):
\begin{enumerate}
\item For a graph $G$, the D2ProperInterval of a graph is the minimum
  number of vertices needed to be removed in order to make $G$ into a
  proper interval graph. For a graph $G$ with D2ProperInterval equal
  $k$, Villanger and van 't Hof \cite{HV13} gave a $6^{k}\poly$-time
  algorithm for finding such a set $S$ to be removed.  To solve \textsc{Independent Set} on a graph $G = (V,E)$, we guess the intersection
  $S'$ of $S$ and an optimal solution, and then combining it with the
  optimal solution of \textsc{Independent Set} on the proper interval
  graph $G- (S \cup N(S'))$.  As \textsc{Independent Set} is solvable in
  $\poly$ time on proper interval graphs, this yields a
  $2^{\bigoh{k}}\poly$ time algorithm.

\item The OCT number of a graph $G$ is the minimum number of vertices
  needed to remove from $G$ in order to make it bipartite. For a graph
  $G$ with OCT number equal $k$, Lokshtanov, Saurabh and
  Sikdar~\cite{lokshtanov2009simpler} gave a $\bigohstar{3^k}$-time
  algorihm for finding the minimum set $S$ of vertices to remove
  from $G$ to make it bipartite. As with the algorithm above, we can
  solve \pname{Independent Set} by guessing the intersection $S'
 $ of $S$ and the optimal solution and then combine it
  with the optimal solution of the bipartite graph $G- (S
  \cup N(S'))$. As \pname{Independent Set} is trivially solvable in
  $\poly$
  time on bipartite graphs, this yields
  a $2^{\bigoh{k}}\poly$ time algorithm.
\end{enumerate}
Note, however, that these parameters are not bounded by treewidth (and
thus also not bounded by clique-width), see Figure~\ref{fig_compare}.


\begin{thebibliography}{10}

\bibitem{BV13}
R.~Belmonte and M.~Vatshelle.
\newblock Graph classes with structured neighborhoods and algorithmic
  applications.
\newblock {\em Theoret. Comput. Sci.}, 511:54--65, 2013.

\bibitem{Bertossi1986}
A.~A. Bertossi.
\newblock Dominating sets for split and bipartite graphs.
\newblock {\em Inform. Process. Lett.}, 19(1):37--40, 1984.

\bibitem{BLRV10}
H.~Bodlaender, E.~van Leeuwen, J.~van Rooij, and M.~Vatshelle.
\newblock Faster algorithms on clique and branch decompositions.
\newblock In {\em Proceedings of MFCS}, volume 6281 of {\em LNCS}, pages
  174--185. Springer, 2010.

\bibitem{BJ1982}
K.~S. Booth and J.~H. Johnson.
\newblock Dominating sets in chordal graphs.
\newblock {\em SIAM J. Comput.}, 11(1):191--199, 1982.

\bibitem{BTV09_hjoin}
B.-M. Bui-Xuan, J.~A. Telle, and M.~Vatshelle.
\newblock {$H$}-join decomposable graphs and algorithms with runtime single
  exponential in rankwidth.
\newblock {\em Discrete Appl. Math.}, 158(7):809--819, 2010.

\bibitem{DBLP:journals/tcs/Bui-XuanTV11}
B.-M. Bui-Xuan, J.~A. Telle, and M.~Vatshelle.
\newblock Boolean-width of graphs.
\newblock {\em Theor. Comput. Sci.}, 412(39):5187--5204, 2011.

\bibitem{BTV13}
B.-M. Bui-Xuan, J.~A. Telle, and M.~Vatshelle.
\newblock Fast dynamic programming for locally checkable vertex subset and
  vertex partitioning problems.
\newblock {\em Theoret. Comput. Sci.}, 511:66--76, 2013.

\bibitem{DBLP:conf/soda/CaoM14}
Y.~Cao and D.~Marx.
\newblock Interval deletion is fixed-parameter tractable.
\newblock In {\em Proceedings of the Twenty-Fifth Annual ACM-SIAM Symposium on
  Discrete Algorithms}, pages 122--141, 2014.

\bibitem{CR05}
D.~G. Corneil and U.~Rotics.
\newblock On the relationship between clique-width and treewidth.
\newblock {\em SIAM J. Comput.}, 34(4):825--847 (electronic), 2005.

\bibitem{courcelle2000linear}
B.~Courcelle, J.~A. Makowsky, and U.~Rotics.
\newblock Linear time solvable optimization problems on graphs of bounded
  clique-width.
\newblock {\em Theory Comput. Syst.}, 33(2):125--150, 2000.

\bibitem{CO2000}
B.~Courcelle and S.~Olariu.
\newblock Upper bounds to the clique width of graphs.
\newblock {\em Discrete Appl. Math.}, 101(1--3):77--114, 2000.

\bibitem{downey1999parameterized}
R.~G. Downey and M.~R. Fellows.
\newblock {\em Parameterized complexity}.
\newblock Monographs in Computer Science. Springer-Verlag, New York, 1999.

\bibitem{flum2006parameterized}
J.~Flum and M.~Grohe.
\newblock {\em Parameterized complexity theory}.
\newblock Texts in Theoretical Computer Science. An EATCS Series.
  Springer-Verlag, Berlin, 2006.

\bibitem{GGRW2003}
J.~F. Geelen, B.~Gerards, N.~Robertson, and G.~Whittle.
\newblock Obstructions to branch-decomposition of matroids.
\newblock {\em J. Combin. Theory Ser. B}, 96(4):560--570, 2006.

\bibitem{GR00}
M.~C. Golumbic and U.~Rotics.
\newblock On the clique-width of some perfect graph classes.
\newblock {\em Internat. J. Found. Comput. Sci.}, 11(3):423--443, 2000.

\bibitem{Gurski08}
F.~Gurski.
\newblock A comparison of two approaches for polynomial time algorithms
  computing basic graph parameters.
\newblock {\em CoRR}, abs/0806.4073, 2008.

\bibitem{Jel10}
V.~Jel{\'{\i}}nek.
\newblock The rank-width of the square grid.
\newblock {\em Discrete Appl. Math.}, 158(7):841--850, 2010.

\bibitem{KR12}
M.~M. Kant{\'e} and M.~Rao.
\newblock The rank-width of edge-coloured graphs.
\newblock {\em Theory Comput. Syst.}, 52(4):599--644, 2013.

\bibitem{LMS11}
D.~Lokshtanov, D.~Marx, and S.~Saurabh.
\newblock Known algorithms on graphs of bounded treewidth are probably optimal.
\newblock In {\em Proceedings of the {T}wenty-{S}econd {A}nnual {ACM}-{SIAM}
  {S}ymposium on {D}iscrete {A}lgorithms}, pages 777--789, Philadelphia, PA,
  2011. SIAM.

\bibitem{LMS11eth}
D.~Lokshtanov, D.~Marx, and S.~Saurabh.
\newblock Lower bounds based on the exponential time hypothesis.
\newblock {\em Bulletin of the {EATCS}}, 105:41--72, 2011.

\bibitem{lokshtanov2009simpler}
D.~Lokshtanov, S.~Saurabh, and S.~Sikdar.
\newblock Simpler parameterized algorithm for {OCT}.
\newblock In {\em Combinatorial algorithms}, volume 5874 of {\em Lecture Notes
  in Comput. Sci.}, pages 380--384. Springer, Berlin, 2009.

\bibitem{niedermeier2006invitation}
R.~Niedermeier.
\newblock {\em Invitation to fixed-parameter algorithms}, volume~31 of {\em
  Oxford Lecture Series in Mathematics and its Applications}.
\newblock Oxford University Press, Oxford, 2006.

\bibitem{Oum2006}
S.~Oum.
\newblock Approximating rank-width and clique-width quickly.
\newblock {\em ACM Trans. Algorithms}, 5(1):Art. 10, 20, 2008.

\bibitem{Oum2008}
S.~Oum.
\newblock Rank-width is less than or equal to branch-width.
\newblock {\em J. Graph Theory}, 57(3):239--244, 2008.

\bibitem{oum2006approximating}
S.~Oum and P.~Seymour.
\newblock Approximating clique-width and branch-width.
\newblock {\em J. Combin. Theory Ser. B}, 96(4):514--528, 2006.

\bibitem{OS2005}
S.~Oum and P.~Seymour.
\newblock Testing branch-width.
\newblock {\em J. Combin. Theory Ser. B}, 97(3):385--393, 2007.

\bibitem{robertson1991graph}
N.~Robertson and P.~Seymour.
\newblock Graph minors. {X}. {O}bstructions to tree-decomposition.
\newblock {\em J. Combin. Theory Ser. B}, 52(2):153--190, 1991.

\bibitem{telle1993practical}
J.~A. Telle and A.~Proskurowski.
\newblock Practical algorithms on partial $k$-trees with an application to
  domination-like problems.
\newblock In {\em Proceedings of the Third Workshop on Algorithms and Data
  Structures}, pages 610--621. Springer-Verlag, 1993.

\bibitem{DBLP:journals/siamdm/TelleP97}
J.~A. Telle and A.~Proskurowski.
\newblock Algorithms for vertex partitioning problems on partial $k$-trees.
\newblock {\em SIAM J. Discrete Math.}, 10(4):529--550, 1997.

\bibitem{HV13}
P.~van~'t Hof and Y.~Villanger.
\newblock Proper interval vertex deletion.
\newblock {\em Algorithmica}, 65(4):845--867, 2013.

\bibitem{martinthesis}
M.~Vatshelle.
\newblock {\em New width parameters of graphs}.
\newblock PhD thesis, University of Bergen, May 2012.
\newblock \url{http://hdl.handle.net/1956/6166}.

\end{thebibliography}
\end{document}